\documentclass[submission,copyright,creativecommons]{eptcs}
\providecommand{\event}{GASCom 2026} 

\usepackage{iftex}

\ifpdf
  \usepackage{underscore}         
  \usepackage[T1]{fontenc}        
\else
  \usepackage{breakurl}           
\fi
\usepackage{tikz} 
\usetikzlibrary{arrows}
\usepackage{amsmath,amsthm}
\usepackage{todonotes,cancel}
\usepackage{subcaption}
\usepackage{makecell}
\usepackage{soul}
\usepackage{refcount}
\usepackage{xspace}
\usepackage{amsfonts,amssymb}
\newcommand{\ie}{\emph{i.e.}\@\xspace}
\makeatletter
\newcommand{\stackover}{\genfrac{.}{.}\z@{}}
\makeatother

\DeclareMathOperator{\lcm}{lcm}
\newcommand{\zo}{\set{0,1}}
\newcommand{\structure}[1]{\ensuremath{\left\langle#1\right\rangle}\xspace}

\newcommand{\set}[1]{\ensuremath{\left\{#1\right\}}\xspace}
\newcommand{\N}{\ensuremath{\mathbb{N}}\xspace}
\newcommand{\Z}{\ensuremath{\mathbb{Z}}\xspace}
\newcommand{\ZO}{\ensuremath{\set{0,1}}\xspace}
\newcommand{\cvres}[2]{\ensuremath{\begin{bmatrix}#1\\#2\end{bmatrix}}\xspace}
\newcommand{\cz}{\ensuremath{\underline{0}}\xspace}
\newcommand{\cu}{\ensuremath{\underline{u}}\xspace}
\newcommand{\cau}{\ensuremath{\underline{1}}\xspace}

\newcommand{\fcxv}[1]{\ensuremath{f\left(#1\right)}\xspace}
\newcommand{\ftcxv}[2]{\ensuremath{f^{#1}\left(#2\right)}\xspace}
\newcommand{\RESOCA}{RESOCA\xspace}
\newtheorem{question}{Question}

\newtheorem{proposition}{Proposition}

\newtheorem{theorem}{Theorem}
\newtheorem{corollary}{Corollary}
\newtheorem{conjecture}{Conjecture}
\newtheorem{lemma}{Lemma}
\newtheorem{remark}{Remark}
\newcommand{\LF}{\ensuremath{\delta}\xspace}
\newcommand{\GF}[1]{\ensuremath{G_{#1}}\xspace}
\title{The Curious Case of Reversible Elementary Second Order Cellular Automaton $115$}
\author{Enrico Formenti
\institute{Universit\'e C\^ote d'Azur, CNRS, i3s\\ Nice, France}
\email{enrico.formenti@univ-cotedazur.fr}
\and
Supreeti Kamilya
\institute{Department of Computer Science and Engineering,\\ Birla Institute of Technology, Mesra, Ranchi, India}
\email{supreeti.kamilya17@bitmesra.ac.in}
}
\newcommand{\titlerunning}{The Curious Case of RESOCA $115$}
\newcommand{\authorrunning}{E. Formenti \& S. Kamilya}

\hypersetup{
  bookmarksnumbered,
  pdftitle    = {\titlerunning},
  pdfauthor   = {\authorrunning},
  pdfsubject  = {EPTCS},               
}
\begin{document}
\maketitle

\begin{abstract}
We prove that the reversible elementary second order cellular automaton rule $115$
is periodic when started on finite initial configurations. We also study some families of finite configurations that have interesting period functions.
\end{abstract}
\section{Introduction}
Cellular automata (CA) are a formal model widely
used in the simulation of complex systems
with application domains ranging from Biology
to Mathematics and Computer Science. 

CA can be seen as a collection of finite automata
arranged on a regular grid. Each automaton
updates its state on the basis of its own state and of the one of a fixed finite set of
neighboring automata. All automata have the same transition function and update their
state in parallel.

Higher order Cellular Automata (HOCA) have
been introduced to include in the local rule
time dependencies for modeling time dependent phenomena. The dynamical behavior
of the class of CA is mostly still unexplored. Only a few results are known and they concern injectivity and surjectivity
properties~\cite{bruyn1991,kari2000}; or a specific subclass~\cite{DennunzioFMMP19}. More recently, a systematic exploration
of HOCA in which memory has been fixed to
$2$ (\ie they remember their previous state)
has started in~\cite{formenti2024elementary}. From one end, one may remark that these systems can be seen as standard CA over a larger alphabet (see~\cite{DennunzioFMMP19}) but
adopting the larger alphabet hides all details about the micro-behavior which might be quite interesting in applications.
From that preliminary explorations the authors came up with the class of reversible elementary second order CA (RESOCA) and more specifically on the one based on the elementary rule $115$. This has absolutely puzzling behavior which might result useful
in practical applications. The main contribution of this paper is that all finite configurations have periodic behavior. However, the function which relates an initial (finite) configuration with its period seems to be quite complex and it is not fully understood yet. We show some families of initial configurations that have simple periodic functions and one which seems to have a quite complex periodic function with exponential growth with respect to the size of the initial
configuration.
\smallskip

The paper is structured as follows. The next section introduces all the main notions
used in this paper. Section~\ref{RESOCA115}
discusses the behavior of RESOCA $115$
and the many questions issued by it.
The following section contains the proof of results in Section~\ref{RESOCA115}.
The \emph{merge construction}, one of the main
tools used to generate new periods from
finite configurations is briefly discussed
in Section~\ref{sec:merge-construction}.
Section~\ref{sec:families} studies some families of initial (finite) configurations and their period. In the last section, we draw our conclusions and propose some perspectives.
\medskip

Due to space constraints, most proofs and tables have been omitted. They will be included in the journal version of this work.

\section{Preliminaries}
\subsection{Elementary Higher Order CA}
An \textbf{elementary higher order CA} (EHOCA) consists of a tuple \structure{S,m,r,\delta} where $S=\zo$ is the
\textbf{set of states} of the CA; $m\in\N$ is the \textbf{memory} (also called the \textbf{order}); $r\in\N$ is
the
\textbf{radius} and accounts for how many left and right neighbors are taken into consideration for updating the state of the current cell. Finally, $\delta\colon\left(\zo^{2r+1}\right)^m\to\left(\zo\right)^m$ is the transition function (also called \textbf{local rule}) that updates the state of the current cell according to
the content of the neighborhood and of the memory. 

A local function \LF induces a \textbf{global function}
$\GF{\delta}\colon\left(\zo^\Z\right)^m\to\left(\zo^\Z\right)^m$
that demonstrates the evolution of the current configuration in the next time step. An elementary second order CA sets the value of $m$ to $2$. More precisely, for elementary second order CA, $\GF{\delta}$ associates
the vector $\Vec{e}=(e^1,e^2)\in\left(\zo^\Z\right)^2$ with the vector $\GF{\delta}\left(\Vec{e}\right)\in\left(\zo^\Z\right)^2$ 
where $\GF{\delta}\left(\Vec{e}\right)^1_i=e^2_i$ and
\[
\GF{\delta}\left(
\Vec{e}
\right)^2_i
=
\delta\left(
\begin{bmatrix}
(e^1_{i-1},e^1_i,e^1_{i+1})\\
(e^2_{i-1},e^2_i,e^2_{i+1})
\end{bmatrix}
\right)
\]
for all $i\in\Z$. In other words, the new vector at position $i$ is obtained by
putting the old value of component $2$ to component one, and
computing the updated value of the new component $2$ using the local rule 
$\delta$.

An EHOCA is \textbf{injective} if its global transition function $\GF{\delta}$ is one-to-one. If the function is onto, the EHOCA is \textbf{surjective}. An EHOCA is \textbf{bijective} if $\GF{\delta}$ is both one-one and onto. We recall that
for any EHOCA, there exists a standard CA which is topologically conjugated
to it~\cite{DennunzioFMMP19}. Hence, we call
a EHOCA \textbf{reversible} if and only if its conjugated CA is reversible.
By a result in~\cite{hedlund1969}, the class of reversible EHOCA coincides
with the class of bijective EHOCA. 

The number of local rules of second order CA is quite huge. It goes beyond the scopes of this paper the dynamics all of them
in a single paper. In this work, we propose to focus our attention on a special subclass of EHOCA rules $\delta_n'$
that are \textbf{induced by an elementary cellular automata rule $\boldsymbol\delta_{\boldsymbol n}$},
more precisely :
\begin{equation}
\label{eq:EHOCA-subset}
\delta'_n\left(
\begin{bmatrix}
    a,b,c\\
    d,e,f
\end{bmatrix}
\right)=
\begin{bmatrix}
    e\hfill\\
b+\delta_n(d,e,f)\mod 2
\end{bmatrix}
\end{equation}
for all $a,b,c,d,e,f\in\zo$.
We call this special class of second order elementary CA as \textbf{\RESOCA}~\cite{formenti2024elementary} defined by Equation~\eqref{eq:EHOCA-subset}.

A \textbf{configuration} is a function from \Z to \ZO. A configuration is \textbf{finite} iff only a finite number of cells have a state different from $0$. For a finite configuration $c$, define $\min c=\min\set{i\in\Z\;|\;c_i=1}$
and $\max c=\max\set{i\in\Z\;|\;c_i=1}$. The
\textbf{length} of $c$ is denoted $|c|$ and is
given by $|c|=|\max c-\min c|+1$.
A configuration $\vec c=\cvres{b}{c}$ is \textbf{finite} if both $b$ and $c$ are finite.
For $\vec c=\cvres{b}{c}$ define $\min\vec c=\min(\min b,\min c)$ and $\max\vec c=\max(\max b,\max c)$, the length of $\vec c$
is given by
$|\vec c|=|\max\vec c-\min\vec c|+1$.

In this paper, we are going to make a large usage of some specific finite configurations for which we are going to introduce a specific notation. We denote \cz a configuration in which all cells are in state $0$; \cu a configuration in which all states are in state $0$ except for the cell at position $0$ that is in state $1$.

\subsection{Elementary Cellular Automata}
\label{sec:ECA}
An EHOCA is an \textbf{elementary CA} (or ECA) if memory $m=1$. A cell of an elementary CA changes its state based on the states of itself, its left neighbor and its right neighbor. Hence,
the local rule $\delta$ can be conveniently represented in tabular form (see for example Table~\ref{table:rule115}). 

\begin{table}[]
\centering
\begin{tabular}{lllllllll}
\hline
    & 111 & 110 & 101 & 100 & 011 & 010 & 001 & 000 \\ \hline
$f$ &\phantom{0}0   &\phantom{0}1   &\phantom{0}1   &\phantom{0}1   &\phantom{0}0   &\phantom{0}0   &\phantom{0}1   &\phantom{0}1   \\ \hline
\end{tabular}
\caption{The tabular representation of the ECA rule 115.}
\label{table:rule115}
\end{table}

Interpreting the output row of the tabular 
representation
of a local rule as a decimal number written in binary gives the well-known
\textbf{Wolfram's number}. Indeed, in the sequel of this paper we will refer to
a specific ECA by its Wolfram's number $n$ and refer to its local rule by $\delta_n$. 

There are total $256$ distinct ECA rules. The rules exhibit different dynamical behaviour, from simple to chaotic. 
The ECA rules has been intensively investigated and several classifications of their dynamics have appeared over the years~\cite{Braga1995,kurka1997languages,culik1988undecidability,schule2012full,EnricoTCS}. 

These $256$ rules are the basis for forming the second order CA rules studied in~\cite{formenti2024elementary}.
During the previously mentioned study we came across the rule $115$, which has a very curious behavior that motivated the present work.

\section{The RESOCA $115$}
\label{RESOCA115}

Given a RESOCA $f$, a configuration \cvres{m}{c} is \textbf{periodic} iff there exists an integer $p$ such that $f^p\left(\cvres{m}{c}\right)=\cvres{m}{c}$; the smallest $p$ satisfying the previous property is the \textbf{period} of \cvres{m}{c}. We say that $p$ is a \textbf{period for} \RESOCA $115$ if there exists a finite configuration which has period $p$ for \RESOCA $115$.

\begin{theorem}
\label{th:all-finite-cfg-are-periodic}
    Any finite configuration is periodic for RESOCA 115.
\end{theorem}

The proof of the previous theorem is long and a bit tedious, and we leave it for the next section.
The fact that RESOCA $115$ has a periodic behavior on finite configurations should not mislead us into thinking that the dynamics is very simple. None of this is true. We will see that there exist classes of configurations which
present very complex sequences of periods, most of which are not yet fully understood. Let us start exploring some conditions that periods should satisfy.

\begin{lemma}
\label{lem:nec-cond-period}
For RESOCA 115,
every periodic finite configuration has a period $p$ such that $p\geq 3$ and $p$ is a multiple of $3$.
\end{lemma}

We call \textbf{background} the particular sequence of states that we can see in a column of the space time diagram of the configuration \cvres{\cz}{\cz} used in the proof of the previous lemma: $[001001001001\dots]^T$ (remark that all the columns are the same for this particular space-time diagram).
Therefore, what the previous lemma
says is that any finite configuration
in a sense synchronizes his period
according to the background, forcing
it to be a multiple of $3$. Hence,
the following is a natural question.

\begin{question}
\label{question:are-all-n-congruent-to-0-periods}
For \RESOCA $115$, in the context of finite configurations, are all multiple of $3$ periods?
\end{question}
Remark that the product of two integers that are multiple of $3$ is a multiple of $3$. This means that if $p$ is a period, then $p^2$ respects the necessary conditions given in Lemma~\ref{lem:nec-cond-period}. Hence, the question follows:
\begin{question}
If $p$ is a period for \RESOCA $115$, is $p^2$ also a period?
\end{question}
We can not answer the previous question yet, but we can answer the following close version of it. 
\begin{proposition}
    \label{mergeperiod}
    Consider two finite configurations of
    period $p$ and $q$, respectively.
    Then $\lcm(p,q)$ is also a period. 
    \end{proposition}
The proof of the previous result is obtained via the
\emph{merge construction} that
is explained in Section~\ref{sec:merge-construction}.
In order to find arguments for answering to Question~\ref{question:are-all-n-congruent-to-0-periods} we made several experiments which are reported in the following sections. All experiments were based on configurations of the type \cvres{\cu}{c} where $c$ is a finite configuration.
In this way, many periods were discovered 
: $3, 6, 12, 15, 27, 30, 54, 60, 63, 123, \ldots$ Two periods $p$ and $q$ are \textbf{twins} if $q=p+3$. For example, in the previous sequence the periods $3$ and $6$ are twins as well as $12$ and $15$.
We do not know whether all periods are twins, but twins seem
to satisfy the following.
\begin{conjecture}
Assume that $p$ and $p+3$ are twins periods. Then, the next twins pair is $p'$ and $p'+3$, where $p'$ is given by
\[
p'=2p+\frac{9-3(-1)^p}{2}
\enspace.
\]
\end{conjecture}
We cannot prove the conjecture but
a weaker statement holds.
\begin{proposition}
\label{prop:twin-periods-new-period}
If $p$ and $p+3$ are twin periods, then
\[
p'=\begin{cases}
2p &\text{if $p$ is odd}\\
2p+6&\text{otherwise.}
\end{cases}
\]
is also a period.
\end{proposition}
We remark that the previous proposition allows to produce a new period but not a pair of twin periods. 

Finally, we observe that
since all periods of RESOCA $115$ are
multiples of $3$ then, we can write
$\Pi=3H$, where $\Pi$ is the set of
periods and $H$ is defined as follows \[
H=\set{k\in\N\;|\; 3k\text{ is a period of RESOCA $115$}}.
\]
Hence, the initial question can be reformulated as follows.
\makeatletter
\newcounter{tmp@questionstar}
\newtheorem{innerquestionstar}{Question}
\def\theinnerquestionstar{\arabic{innerquestionstar}*}
\newenvironment{questionstar}[1][]{%
\setcounter{tmp@questionstar}{\value{innerquestionstar}}
\ifx&#1&{}\else\setcounterref{innerquestionstar}{#1}\addtocounter{innerquestionstar}{-1}\fi
\innerquestionstar
}{\endinnerquestionstar\setcounter{innerquestionstar}{\value{tmp@questionstar}}}
\makeatother
\begin{questionstar}[question:are-all-n-congruent-to-0-periods]
Is $H=\N$?
\end{questionstar}

If the answer to the previous question
is negative then, more questions arise.
\begin{question}
Is $H$
syndetic (\ie does it have bounded gaps) or
thick (\ie does it contain arbitrary large intervals)?
\end{question}
Answering to these question is important
not only for the theoretical aspects
but also for the practical aspects
like the application to cryptography
that we propose at the end of this work, for example.
\section{Periodicity of finite configurations}
\label{sec:periodicity-of-finite-cfg}
We are going to decompose the proof of Theorem~\ref{th:all-finite-cfg-are-periodic}
into a series of Lemmata to make it easier to follow.
\begin{lemma}
\label{lem:periodicity-1}
 Consider a finite configuration $b$ such that $b_i=0$ for all $i\leq0$ and another finite configuration $c$ such that $c_i=0$ for all $i<0$ and $c_0=1$. Let $f$ be the global rule of \RESOCA $115$ and $d^t=f^t\left(\cvres{b}{c}\right)_2$. Then we have for all time $t\in\N$
 \[
 \forall i\in\N\setminus\set{0,1}\quad d^t_{-i}=
 \begin{cases}
 1&\text{if\ }t\equiv1\mod3\\
 0&\text{otherwise.}
 \end{cases}
 \]
 and for all times $t\geq2$
 \[
 d^t_{-1}= \begin{cases}
 d^{t-1}_{-1}&\text{if\ }t\equiv0\mod3\\
 1&\text{if\ }t\equiv1\mod3\\
 1-(d^{t-1}_0+d^{t-2}_{-1})&\text{if\ }t\equiv2\mod3
 \end{cases}
 \]
 where addition is taken modulo $2$.
\end{lemma}
\begin{lemma}
\label{lem:periodicity-2}
 Consider two finite configurations $b$ and $c$
 such that, for both of them, all cells with an index strictly larger than $0$ are in state $0$. Let $f$ be the global rule of \RESOCA $115$ and $d^t=f^t\left(\cvres{b}{c}\right)_2$. Then we have, for all times $t\in\N$
 \[
 d^t_1=
 \begin{cases}
 1&\text{if\ }t\equiv1\mod3\\
 0&\text{otherwise.}
 \end{cases}
 \]
\end{lemma}
In other words, the previous lemma says that if we have an initial condition as described, then the column with index $1$ will show only the background.
\begin{remark}
Assume that $b$ and $c$ are two finite configurations satisfying the hypothesis
of the previous Lemma~\ref{lem:periodicity-2}. Then, $\sigma(b)$ and $\sigma(c)$ also satisfy the hypothesis. Hence, we can conclude that \cvres{b}{c} will show background for all columns with index bigger than
$1$.
\end{remark}
\begin{lemma}
\label{lem:periodicity-3}
 Consider two finite configurations $b$ and $c$ such that $b_i=c_i=0$ for all $i<0$ and $b_0=c_0=1$. Let $f$ be the global rule of \RESOCA $115$ and $d^t=f^t\left(\cvres{b}{c}\right)_2$. Then we have, for all time $t\in\N$
 \[
 \forall i\in\N\setminus\set{0}\quad d^t_{-i}=
 \begin{cases}
 1&\text{if\ }t\equiv1\mod3\\
 0&\text{otherwise.}
 \end{cases}
 \]
 and for all times $t\geq2$
 \[
 d^t_0= \begin{cases}
 d^{t-1}_0&\text{if\ }t\equiv0\mod3\\
 1&\text{if\ }t\equiv1\mod3\\
 1-(d^{t-1}_1+d^{t-2}_0)&\text{if\ }t\equiv2\mod3
 \end{cases}
 \]
 where addition is taken modulo $2$.
\end{lemma}
\begin{lemma}
\label{lem:periodicity-4}
  Consider two finite configurations $b$ and $c$ such that $b_i=c_i=0$ for all $i<0$ and $b_0=c_0=1$. Let $f$ be the global rule of \RESOCA $115$ and $d^t=f^t\left(\cvres{b}{c}\right)_2$. Then we have, for all times $t\geq3$
 \[
 \forall i\in\N\setminus\set{0,1}\quad d^t_{-i}=
 \begin{cases}
 1&\text{if\ }t\equiv1\mod3\\
 0&\text{otherwise.}
 \end{cases}
 \]
 and
 \[
 d^t_{-1}= \begin{cases}
 d^{t-1}_0&\text{if\ }t\equiv0\mod3\\
 1&\text{if\ }t\equiv1\mod3\\
 1-(d^{t-1}_1+d^{t-2}_0)&\text{if\ }t\equiv2\mod3
 \end{cases}
 \]
 where addition is taken modulo $2$.
\end{lemma}
\begin{proof}[Proof of Theorem~\ref{th:all-finite-cfg-are-periodic}]
Consider two finite configurations $b$ and $c$. If $b=c=\cz$, then it is not difficult to see that the space-time diagram of \cvres{\cz}{\cz} contains only background columns and hence has period $3$. Now assume that at least one
of $b$ or $c$ is different from \cz.
Without loss of generality we can assume that for all $i<0$, $c_i=b_i=0$
(otherwise we can choose $c'=\sigma^u(c)$ and $b'=\sigma^u(b)$, where $u=\min \cvres{b}{c}$).
By Lemma~\ref{lem:periodicity-2}, all
columns with index larger than $\max \cvres{b}{c}$ are background.
We will distinguish three disjoint cases. For each of them, we will prove
that for all $i<-1$, all columns in the space-time diagram of \cvres{b}{c} contain only background. This and the previous remark allow to deduce that
the sequence of states that we can see in the space-time diagram of initial condition \cvres{b}{c} between cell $-1$ and cell $\max\cvres{b}{c}$ is ultimately periodic. However, since RESOCA 115
is reversible, we can conclude that the sequence is periodic. Here are the
proofs for the three cases announced above:
1) $b_0=0$ and $c_0=1$: this is Lemma~\ref{lem:periodicity-1};
2) $b_0=1$ and $c_0=0$: this is Lemma~\ref{lem:periodicity-4};
3) $b_0=1$ and $c_0=1$:
    this is Lemma~\ref{lem:periodicity-3}.
\end{proof}

Another interesting consequence of the lemmata in this section is that, for any finite configuration $\vec c$, the size of the configuration and its position are essentially preserved along the evolution. More precisely,
we can give the following.
\begin{corollary}
\label{cor:max-min-finite-cfg}
Let $f$ be the global rule RESOCA $115$. For any finite configuration
$\vec c$, we have
$\forall t\in\N\;\max\vec c=\max f^t(\vec c)$
and
$\forall t\in\N\setminus\set{0}\;\min f(\vec c)=\min f^t(\vec c)
$.
Moreover, $\min f(\vec c)-\min\vec c\leq1$.
\end{corollary}
Although the previous result follows rather directly from earlier results, it is nevertheless of significant importance. Indeed, a result of Moothathu implies that the existence of periodic points with the properties described in the corollary guaranties the existence of infinitely many periodic
points with the same period~\cite{Moothathu2005}. Explicit examples of this fact are provided in Section~\ref{sec:families}.

\section{The merge construction}
\label{sec:merge-construction}
We dedicate to this construction a separated section since it is useful in different situations
when trying to build configurations with peculiar periods.

Consider two finite configurations $\vec c_1, \vec c_2$. By Theorem~\ref{th:all-finite-cfg-are-periodic}, these configurations
are periodic of
period $p$ and $q$, respectively.
Using Lemma~\ref{lem:periodicity-1} through~\ref{lem:periodicity-4},
we remark that along the evolution of
both configurations no modification (with respect to background) is made before $\min\vec c_1-1$ (resp., $\min\vec c_2-1$) or after $\max\vec c_1$ (resp.,
$\max\vec c_2$).
These facts suggest that if we build an initial configuration in which we position the non-background part of $\vec c_1$ on the left of the cell of index $0$ and $\vec c_2$ on the right, leaving a sufficient number of cells in state $0$ separating them, then we should have a new periodic configuration with the period which is some combination of the two original periods. More formally, define $u$ as the largest
integer such that $\max \sigma^u\left(\vec c_1\right)\leq-3$ and $v$ as the least integer such that $\min\sigma^v\left(\vec c_2\right)\geq3$.
Let us build a new finite configuration $\vec c_3$ that has a first component
\[
\forall i\in\Z\quad((\vec c_3)_1)_i=
\begin{cases}
\sigma^u((\vec c_1)_1)_i&\text{if }i\leq0\\
\sigma^v((\vec c_2)_1)_i&\text{if }i>0
\end{cases}
\]
and as a second component
\[
\forall i\in\Z\quad ((\vec c_3)_2)_i=
\begin{cases}
\sigma^u((\vec c_1)_2)_i&\text{if }i\leq0\\
\sigma^v((\vec c_2)_2)_i&\text{if }i>0.
\end{cases}
\]

\begin{proposition}
Let $\vec c_1, \vec c_2$ and $\vec c_3$ as above. Let $p$ (resp., $q$) be the period of $\vec c_1$ (resp., $\vec c_2$). Then, $\vec c_3$ is periodic of period $\lcm(p,q)$.
\end{proposition}
\begin{proof}
The proof is trivial and left to the reader.
\end{proof}

Remark that the merge construction
can be generalized, in the obvious way,
to work with an arbitrary (but finite) number of finite configurations.


\section{Families of configurations and their period functions}
\label{sec:families}

In this section, we focus our attention on families of initial configurations. Since all initial configurations share the same memory $b=\cu$, each family will be described solely in terms of its set of current configurations. Every family is defined by a regular expression depending on a parameter $k\in\N$. For each value of $k$, the corresponding regular expression generates a binary word $w$ of odd length, from which the current current configuration is constructed by setting $c_{-k+i}=w_i$ for
$i\in[0,|w|]$, where $k=\frac{|w|-1}{2}$.
Our objective is to determine the \textbf{period function} $P(k)$ associated with the family, that is, the function assigning the initial configuration of parameter $k$ its corresponding period.

\subsection{Family $10^{2k-1}1$ ($k\geq1$)}
This is a very simple family of configurations as
it is confirmed by the following:
\begin{proposition}
\label{prop:family-10-pow(2k-1)1}
 The period function $P(k)$ of the family $10^{2k-1}$1 is as follows:
\[
P(k) =
\begin{cases}
2^{\,2k+4}-1,&\text{if\ } k=1,2\\[4pt]
15,&\text{otherwise}.
\end{cases}
\]    
\end{proposition}
\subsection{Family $1^{2k+1}$}
Also this family of periods is pretty simple
but we need a different proof technique than
the previous one.

\begin{proposition}
\label{prop:family-1pow(2k+1)}
 The period function $P(k)$ of the family $1^{2k+1}$ is as follows:
 \[
 \forall k\in\N, P(k)=  
 \begin{cases}
 3&\text{if\ }k=0\\
 12&\text{if\ }k=1\\
 6&\text{otherwise.}
 \end{cases}
 \]
\end{proposition}
\subsection{Family $1(01)^k$}
For this family, the period function appears to
be much more complex than the one of the two previous families.
Using a computer program, we produced a long sequence of
periods.
However, we have no clue for a closed
formula for the period function.

Looking more carefully, 
we remarked
that $P(2h+1)=4^{2h+3}-1$. We suspect
that this is true for any $h\in\N$. Hence, we formulate the following.
 
\begin{conjecture}
    \label{conjecture:base4power4}
    The period function $P(k)$ of
    the family $1(01)^k$ is such
    that $P(2i+1)=4^{2i+3}-1$
    for $i\in\N$.
\end{conjecture}

For even values of $k$, the period function $P(k)$ seems more complicate.
However, from our experiments we could
recognize a pattern.
For $i>2$, call $q_i$ the number obtained by adding $1$ to $P(2i)$ and then dividing the result by $8$.
We could remark that the representation in base $4$ of $q_i=(P(2i)+1)/8$ is given by
$1^i0^22^{i-2}3$ (the highest weight is at the left). Therefore, $q_i$
can be written as $q_i=3\cdot 4^0+2\cdot\sum_{j=1}^{i-2}4^j+\sum_{j=i+1}^{2i}4^j=3+\frac{4^{2i+1}-14\cdot4^{i-1}-8}{3}$. Recalling the relation between $q_i$ and $p_i$,
we get $P(2i)=8p_i-1=23+\frac{2\cdot4^{2i+2}-7\cdot4^{i+1}-4^3}{3}=\frac{2\cdot4^{2i+2}-7\cdot4^{i+1}+5}{3}$. We can therefore emit the
following.

\begin{conjecture}
    \label{conjecture:base4power4even}
    The period function $P(k)$ of
    the family $1(01)^k$ is such
    that $P(2i)=\frac{2\cdot4^{2i+2}-7\cdot4^{i+1}+5}{3}$
    for $i\in\N$.
\end{conjecture}
\section{Conclusions and perspectives}
In this paper, we investigated the behavior of RESOCA 115. In particular, we proved that, for every finite initial configuration, the evolution is periodic with period divisible by 3. We also identified families of finite configurations exhibiting either a simple period function or a considerably more intricate one. Nevertheless, the relationship between the structure of a finite initial configuration and the exact value of its period remains poorly understood. Clarifying this connection constitutes a challenging and promising direction for future research.

Through computational experiments, we identified the following periods:
$3, 6, 9^*, 12, 15, 27, 30, 54,\break 60, 63, 123, 126, 135, 189, 246, 252, 255, 270, 315, 378, 507, 510, 615, 630, 1020, 1023, 1230, 1260, 1890,\break 2043, 2046, 2535, 2583$ and many others. We emphasize that every period in this list arises from a finite configuration, with the notable exception of period 9. Indeed, rather curiously, we were only able to construct non-finite configurations exhibiting this period.

This observation prevents the use of the merge construction to fill the gap at period $18$. More precisely, the technique employed to establish the existence of period $9$ does not extend to period $18$, as it would require the analysis of graphs with $2^{72}$ vertices and $2^{74}$ edges; far beyond the capabilities of our machines\footnote{The technique used for period $9$ was suggested by Jarkko Kari in a private communication.}. Developing more efficient computational methods for proving or disproving the existence of a prescribed period therefore constitutes another compelling direction for future research.

All the hype surrounding the periods is that if we succeed in getting a certain amount of knowledge about them, then we could provide a concrete playground for a new secret sharing protocol in which the shares are given by some finite configurations and the secret is hidden in the period of a configuration obtained by the shares via the merge construction.
We highlight several advantages of this new protocol. Suppose there are n shares, where $n$ itself is secret. In this setting, the presence of all shares is not necessarily required to decrypt the secret; however, one is not completely sure that the decryption is correct until every share has been taken into account. Moreover, the size of each share can be chosen to be logarithmic in the bit-size of its period, and the order in which the shares are processed is irrelevant.

A thorough analysis of this protocol is currently under investigation and will be presented in a separate forthcoming paper.

%
\bibliographystyle{eptcs}
\bibliography{biblio.bib}
\end{document}